\documentclass[11pt, a4paper]{article}
\usepackage[utf8]{inputenc} 
\usepackage{lmodern}
\usepackage[T1,T5]{fontenc} 

\usepackage[
	backend=bibtex, 
	style=authoryear-comp,
	sorting=nyt, 
	mincitenames=1,maxcitenames=2, 
	giveninits=true, 
	dashed=false,
	doi=false,isbn=false,url=false,eprint=false,
	natbib=true]
{biblatex} 
\DeclareNameAlias{sortname}{last-first}
\DeclareFieldFormat[article, inbook, incollection, inproceedings, misc, thesis, unpublished]{title}{#1}
\renewbibmacro{in:}{}

\makeatletter
\DeclareFieldFormat{citehyperref}{%
  \DeclareFieldAlias{bibhyperref}{noformat}
  \bibhyperref{#1}}
\savebibmacro{cite}
\savebibmacro{textcite}
\renewbibmacro*{cite}{%
  \printtext[citehyperref]{\restorebibmacro{cite}\usebibmacro{cite}}}
\renewbibmacro*{textcite}{%
  \printtext[citehyperref]{\restorebibmacro{textcite}\usebibmacro{textcite}}}
\makeatother

\usepackage{amsmath}
\usepackage[mathscr]{eucal}
\usepackage{amsfonts}
\usepackage{amsmath,amsxtra,latexsym,amsthm, amssymb, amscd}
\usepackage{graphicx}
\graphicspath{{figs/}{./}}
\usepackage{amsfonts}
\usepackage{amssymb}
\usepackage{tabularx}
\usepackage{rotating}
\usepackage{array}
\usepackage[usenames,dvipsnames]{xcolor}
\definecolor{darkblue}{RGB}{0,0,139}
\usepackage{caption,setspace}
\usepackage{authblk}
\usepackage[margin=1in]{geometry}
\usepackage[flushleft]{threeparttable}
\usepackage{url}
\usepackage{hyperref}
\hypersetup{
	colorlinks = true,
	urlcolor = {darkblue}
}
\usepackage{csquotes}
\usepackage[small]{titlesec}
\usepackage[en-IE]{datetime2}
\usepackage{graphicx}
\usepackage{subcaption}

\usepackage{appendix}
\DTMlangsetup[en-IE]{showdayofmonth=false}
\hypersetup{
	colorlinks=true, 
	citecolor=darkblue,					
	linkcolor=darkblue
}	

\allowdisplaybreaks
\providecommand{\keywords}[1]{\textbf{\textit{Keywords:}} #1}
\providecommand{\jel}[1]{\textbf{\textit{JEL code:}} #1}
\newtheorem{remark}{Remark}
\newtheorem{proposition}{Proposition}

\title{Globalization, economic growth, and innovation: A two-country two-period model\thanks{The authors thank Thanh-Truong Bui for his remarks and suggestions.}
		}
\author[1]{Cuong Le Van}
\author[2]{Duc V. Le}
\author[3]{Thanh Tam Nguyen-Huu\thanks{Corresponding author: \href{mailto:tnguyenhuu@em-normandie.fr}{tnguyenhuu@em-normandie.fr} (Nguyễn Hữu Thành Tâm).}
	\thanks{The other contacts: \href{mailto:Cuong.Le-Van@univ-paris1.fr}{Cuong.Le-Van@univ-paris1.fr} (Lê Văn Cường); 
								\href{mailto:dvl11@georgetown.edu}{dvl11@georgetown.edu} (Lê Vĩnh Đức).}}
\affil[1]{\small IPAG Business School, Paris School of Economics, CNRS, Foreign Trade University}
\affil[2]{\small Universit\'e Paris 1 Panth\'eon-Sorbonne, Georgetown University}
\affil[3]{\small EM Normandie Business School, M\'etis Lab}

\begin{document}

\maketitle

{\footnotesize\noindent \textcopyright\ 2026. This manuscript version is made available under the \href{https://creativecommons.org/licenses/by-nc-nd/4.0/}{CC-BY-NC-ND 4.0 license}. The peer-reviewed and accepted version is published in: \textit{International Economics} \textbf{187} (2026), article 100725. Please cite the published version of record at: \url{https://doi.org/10.1016/j.inteco.2026.100725}.\par}
\medskip

\begin{abstract}
This paper examines how a developing country can benefit from trade liberalization. We develop a two-period model, comprising an autarky phase and a globalization phase, and a two-country framework, featuring a developing country and a developed country (representing the rest of the world). Our findings indicate that globalization may disadvantage a developing country when its total factor productivity (TFP) is significantly lower than that of the developed country. However, we demonstrate that the developing country can still achieve gains from trade openness by allocating part of its capital to innovation during the autarky period, thereby enhancing its TFP.

\end{abstract}

\keywords{Globalization, TFP, Investment in Innovation, Gain.}

\jel{F1, O41, O3.}


\section{Introduction}
Over the past few decades, most countries have liberalized their economies, leading to increased cross-border flows of goods, services, technology, and capital. This phenomenon, commonly referred to as economic globalization or trade liberalization, is governed by international institutions such as the World Trade Organization (GATT/WTO). The conventional wisdom on economic growth and development suggests that open economies grow faster than closed ones \citep{GrossmanHelpman91, Dollar92, Edwards93}. However, neither existing theoretical models nor empirical evidence has reached consensus on the relationship between trade liberalization and economic gains.

Theoretically, traditional views predict that global trade stimulates growth, whereas recent conceptual work argues that this is not always the case. Trade openness is generally pursued with the expectation of promoting growth for both developed and developing countries. For example, studies based on endogenous growth models \citep{BarroSala97,BaldwinEtAl05,GrossmanHelpman90,GrossmanHelpman91,Lucas88,Acemoglu09} assert that global trade expansion can promote growth by facilitating the diffusion of knowledge and technology through imports of advanced products, foreign direct investment (FDI), and the accumulation of human capital through “learning-by-doing.” Nevertheless, these benefits depend on the degree of openness. Furthermore, \citet{AlesinaEtAl00,BondEtAl05} argue that trade openness enlarges market size, enabling economies to take advantage of increasing returns to scale and specialization. Their consensus emphasizes that trade enhances productivity through both short- and long-run effects: alleviating resource misallocation in the short term and driving technological innovation in the long term. In addition, trade liberalization can trigger domestic reforms in response to international competition, reinforcing economic development \citep{SachsWarner95}. Thus, trade policies can influence technological change through investment within endogenous growth frameworks.

In contrast, other theoretical studies suggest that trade openness may hinder growth, depending on the country's development level. For instance, less-developed countries lacking human capital \citep{Abramovitz86} or financial development \citep{AghionEtAl05} may be unable to absorb advanced technologies from developed economies. Endogenous growth models further suggest that if an economy specializes in sectors with a “comparative disadvantage” in productivity growth \citep{Redding99}, or if innovation and learning-by-doing are exhausted \citep{Lucas88,Young91}, trade openness can reduce growth over time. In such cases, “selective protection” combined with technological advancement may be more effective than a full liberalization \citep{Redding99}.

Empirical evidence is equally inconclusive. Some studies find a positive link between trade liberalization and growth \citep{LevineRenelt92,Harrison96,DollarKraay04,ChangEtAl09,Kim11,SachsWarner95,WacziargKaren08}, while others report no association or even negative effects \citep{ClemensWilliamson01,MusilaYiheyis15,Ulasan15,RodriguezRodrik00,Singh10}. These discrepancies may stem from differences in proxy, methodologies, and model specifications. Moreover, trade impacts vary across countries. \citet{KimLin09}, for example, show that openness benefits high-income economies but harms low-income ones, likely due to differences in absorptive capacity for knowledge and technology. 

Against this backdrop, our paper develops a theoretical endogenous growth model to explain how a country can benefit from globalization under specific conditions of trade liberalization. As prior evidence suggests, industrialized economies tend to gain the most from globalization, while the impacts on developing countries remain ambiguous. Our study focuses on developing countries and explores how they can leverage trade openness to enhance growth.

We construct a two-country model comprising a developing country (H) and a developed country (F, representing the rest of the world). Our analysis focuses on productivity and the utility of total consumption in country H under globalization. The framework spans two periods: an autarky phase, in which both countries are closed to trade, and a globalization phase, in which they sign a trade agreement allowing free trade. We examine two scenarios for country H: (1) allocating all initial capital to production in the autarky period, and (2) reserving part of its capital for innovation to improve total factor productivity (TFP) in the globalization phase. The model assumes the Law of One Price (LOP), which states that identical goods trade at the same price in the absence of frictions.\footnote{LOP holds under free competition, price flexibility, and no trade frictions such as tariffs or transport costs.}

Our findings indicate that if country H invests all its capital in production during autarky, it may lose out under globalization when its TFP is significantly lower than that of country F. In contrast, if H allocates part of its initial capital to innovation, its TFP improves in the second period, enabling it to gain from globalization.

This research makes two main contributions to the trade-growth literature. First, it shows that the growth effects of trade liberalization are conditional on pre-liberalization innovation effort and initial productivity levels. Globalization may fail to enhance growth when a developing country enters liberalization with a large productivity gap relative to the rest of the world. Thus, unlike standard endogenous growth models in which openness is generally growth-enhancing, our framework highlights that the effects of trade liberalization are not uniform, but depend on whether a country has accumulated sufficient capabilities before opening up. Second, our model identifies a threshold level of innovation investment above which innovation-led improvements in TFP enable trade liberalization to generate higher growth. This threshold-based mechanism provides a theoretical explanation for the heterogeneous growth outcomes of globalization in developing and advanced economies. More broadly, the paper offers a framework for understanding why the effects of openness may differ systematically between countries at different stages of development.

The remainder of our work is organized as follows. In \hyperref[literature]{Section~\ref*{literature}}, we review the theoretical and empirical literature according to trade liberalization and economic growth. \hyperref[model]{Section~\ref*{model}} basically develops a two-country model. We also examine how to gain from trade liberalization investment in innovation towards improving the country's TFP. \hyperref[con]{Section~\ref*{con}} concludes our main findings and potential research in the future.

\section{Literature Review}\label{literature}

The standard doctrines of international trade typically predict that countries benefit from openness through higher economic growth. These gains arise through specialization, investment in technological innovation, productivity improvements, and complementary policies such as better resource allocation and improved education. This line of reasoning originates with the classical theory of \citet{Ricardo17}, who argued that international trade enables the allocation of scarce resources more efficiently. Later, the landmark contribution of \citet{Solow57}, based on the neoclassical growth model, identified technological change as an exogenous driver of economic growth. Historically, Solow’s work has become one of the most influential findings in growth economics, demonstrating the crucial role of total factor productivity (TFP). His 1957 study showed that gross output per hour worked in the United States doubled between 1909 and 1949, with nearly seven-eighths of this increase attributable to technological progress and only one-eighth to capital deepening. These results were confirmed by \citet{Kendrick61}, who estimated that TFP accounted for 80.0 and 88.5 percent of the labor productivity's growth during the periods 1869–1953 and 1909–1948, respectively. More recent studies have reinforced these findings by emphasizing TFP as the dominant source of long-run growth (e.g., \citealp{Gordon99, Gordon16}).

However, the original Solow framework did not explicitly consider the role of policies such as education or trade in driving growth. Subsequent research addressed this limitation. For example, \citet{NelsonPhelps66} and \citet{Schultz75} showed that investments in education increase the supply of skilled labor and thus improve a country’s ability to adopt and exploit new technologies. Building on this idea, \citet{Lucas88} developed an endogenous growth model in which human capital accumulation—driven by learning-by-doing and on-the-job training—plays a central role in linking international trade with technological progress and productivity growth. Similarly, \citet{GrossmanHelpman91} modeled endogenous technological change and showed that greater trade openness facilitates the diffusion of frontier technologies, thereby promoting productivity and long-term growth. Although the theoretical literature generally supports the growth-enhancing effects of trade liberalization, substantial debate remains, particularly in empirical work.

In particular, several theoretical contributions suggest that the impact of trade liberalization on economic performance is not always uniformly positive. Although trade openness can improve productivity and growth through channels such as knowledge spillovers, technology diffusion, and innovation \citep{Lucas88, GrossmanHelpman91}, the ability to benefit from these mechanisms depends critically on a country’s learning and absorptive capacity. This capacity is influenced by financial development \citep{AghionEtAl05}, firms’ investment in R\&D \citep{CohenLevinthal90}, human capital \citep{Wozniak87, Cosar09}, foreign direct investment \citep{Findlay78}, and governance quality \citep{McmillanEtAl14}. Moreover, excessively rigid regulations can impede growth by preventing resources from being reallocated toward more productive firms and industries \citep{BolakyFreund08}. As a result, developing countries that lack R\&D, skilled labor, effective financial systems, or supportive institutions may fail to fully benefit from globalization.
Empirical evidence mirrors these theoretical ambiguities. Early studies, consistent with traditional theory, found positive growth effects of trade liberalization. For instance, \citet{SachsWarner95} reported that more open economies grew 2 to 2.5 percentage points faster than closed ones. Using updated data for 1950–1998, \citet{WacziargKaren08} showed that countries experienced growth rates about 1.5 percentage points higher after liberalization.

By contrast, later research has questioned whether trade liberalization always delivers growth gains. \citet{RodriguezRodrik00} raised serious doubts about the robustness of earlier results. Using the so-called “tariff–growth paradox,” \citet{ClemensWilliamson01} found a negative relationship between openness and growth, showing that protection was associated with faster growth before World War II but slower growth afterward. A survey by \citet{Singh10} similarly concluded that the growth effects of outward-oriented trade are far from uniform, while \citet{Tekin12} found no significant impact of trade liberalization in a sample of 27 least-developed countries in Africa.

The effects of trade openness also appear to depend on timing, degree of openness, and country characteristics. \citet{GreenwayEtAl02} documented a J-curve pattern in developing countries: short-run effects may be harmful as resources exit sectors with comparative disadvantage, but long-run growth improves as resources are reallocated toward more productive activities. Likewise, \citet{Pam18} found an inverted U-shaped relationship between trade openness and growth in 42 sub-Saharan African countries, implying that openness boosts growth up to a threshold, beyond which its marginal effect declines.

Furthermore, the impact of trade liberalization varies with a country’s level of development. Several studies find positive effects for advanced economies but weaker or even adverse effects for developing ones \citep{KimLin09, Herzer13}. This pattern is consistent with the regional asymmetries documented by \citet{ClemensWilliamson01}, who showed that tariffs were positively associated with growth in rich countries but negatively or weakly related in poorer economies. \citet{KimLin09} identified an income threshold above which trade globalization promotes growth and below which it may be harmful.

Finally, trade liberalization alone is not sufficient to guarantee strong economic performance. While greater openness can stimulate growth, it often operates through complementary channels such as capital accumulation and foreign investment. For example, \citet{LevineRenelt92} showed that trade liberalization encourages foreign direct investment by reducing tariffs, thereby supporting long-term growth. \citet{Harrison96} likewise found that the positive effect of openness works largely through capital accumulation. More recently, \citet{TrejosBarboza15} argued that trade openness was not the primary driver of Asia’s post-1950 growth miracle. Instead, productivity gains associated with capital deepening played a more central role. In general, technological innovation, productivity, and economic growth are mutually reinforcing, leading to sustained long-run development \citep{HallJones99, Rouvinen02}.
 
Accordingly, our study contributes to the trade-growth literature by developing a new theoretical framework that highlights the role of innovation investment. In particular, we demonstrate the existence of threshold effects under which trade liberalization can generate economic growth.


\section{The two-country model}\label{model}

We consider two countries: a developing country $H$ and a developed country $F$, which represents the rest of the world. The two countries live for two periods, $t=0,1$ (referred to as the first and second periods, respectively). The economic environment is as follows:
\begin{itemize}
\item[(i)] In period $0$, the two countries are in autarky; there is no trade between them.
\item[(ii)] In period $1$, $H$ and $F$ sign a trade agreement that allows for free trade and capital mobility between the two countries.
\end{itemize}

Each country is endowed with an initial stock of capital: country $H$ holds $k_{H,0}$, and country $F$ holds $k_{F,0}$. Preferences are identical across countries and are represented by the intertemporal utility function

\begin{eqnarray}\label{e_U}
    U(c_{i,0}, c_{i,1}) = \ln c_{i,0} + \ln c_{i,1}, \qquad i \in \{H,F\},
\end{eqnarray}
where $c_{i,t}$ denotes consumption in country $i$ in period $t$.

In each country, there is a representative consumer who maximizes the above intertemporal utility function. Assume also that there is a representative firm in each country that uses capital to produce a consumption good. Production technologies are given by
\[
Y_H = A_H k_H^{\alpha}, \qquad 
Y_F = A_F k_F^{\alpha},
\]
where $\alpha \in (0,1)$, $A_i$ denotes TFP in country $i$, and $k_i$ is the amount of capital employed in production.\footnote{In such a context, $c_{i,t}, k_{i,t}$ measure the consumption and capital per capita of country $i$ in period $t$.} 

\subsection{Autarky}
Consider the first period in which $ H$ and $ F$ are in autarky. The representative  consumer in country $i$ ($i=H,F$) solves:
\begin{eqnarray*}
&&\max \{\ln c_{i,0} +\ln c_{i,1}\}\\
&&\hbox{s.t. } c_{i,0} + k_{i,1} = A_i k_{i,0}^\alpha\\
&&c_{i,1}= A_i {k_{i,1}}^\alpha
\end{eqnarray*}
We obtain:
$$k_{is,1} = \frac{\alpha}{1+\alpha}A_i k_{i,0}^\alpha $$
where $k_{is,1}$ is the supply of capital of the consumer.
\\[5mm]
The firm of country $i$ maximizes its profit in order to obtain its demand for capital $k_{id,1} $:
\begin{eqnarray*}
\max \{p_i A_i k_{i,1}^\alpha - \rho_i  k_{i,1}\}
\end{eqnarray*}
here, $p_i, \rho_i$ are respectively the price of consumption good and the price of capital in country $i$. Let $r_i = \frac{\rho_i  }{p_i}$ denote the real interest rate in country $i$. The demand for capital of the firm in country $i$ is:
\begin{eqnarray*}
k_{id,1}= \left(\frac{A_i \alpha}{r_i} \right)^{\frac{1}{1-\alpha}}
\end{eqnarray*}
The real interest rate clears the capital market, i.e.
\begin{eqnarray*}
k_{is,1}&=& k_{id,1} 
\Leftrightarrow \frac{\alpha}{1+\alpha}A_i k_{i,0}^\alpha=  \left (\frac{A_i \alpha}{r_i} \right )^{\frac{1}{1-\alpha}} 
\end{eqnarray*}
We obtain:
\begin{eqnarray}\label{r^H}
 r_i= \frac{ {A_i}^\alpha \alpha ^\alpha (1+\alpha)^{1-\alpha}}{k_{i,0}^{\alpha(1-\alpha)}} 
 \end{eqnarray}
The maximum profit of country $i$ under autarky is:
\begin{eqnarray}\label{pi^H autarky}
\Pi _i = \left (\frac{A_i}{{r_i} ^\alpha} \right ) ^ {\frac{1}{1-\alpha}} (\alpha ^{\frac{\alpha}{1-\alpha}}- \alpha ^{\frac{1}{1-\alpha}}) 
\end{eqnarray}
\begin{remark}
We should, normally, write the intertemporal utility as: $$\ln c_{i,0} +\beta\ln c_{i,1}$$ where $\beta$ is the discount factor. Consequently, we obtain the supply of capital in country $i$ as: $$k_{is,1} = \frac{\alpha\beta}{1+\alpha\beta}A_i k_{i,0}^\alpha $$
from where we can compute the equilibrium real interest rate under the autarky in each country as:
\begin{eqnarray}\label{e_beta}
r_i= \frac{ {A_i}^\alpha \alpha ^\alpha (1+\alpha\beta)^{1-\alpha}}{\beta^{1-\alpha}k_{i,0}^{\alpha(1-\alpha)}}    
\end{eqnarray}
Notice that the real interest rates in Equations (\ref{r^H}) and (\ref{e_beta}) have similar characteristics (e.g, they increase with the productivity $A_i$ and decrease with the initial endowment $k_{i,0}$). For the sake of simplicity, we assume $\beta=1$ and use the utility function given in Equation (\ref{e_U}) in this research. Since we do not intend to investigate how the discount factor affects the economy under globalization, this assumption has no impact on the main paper's findings.
\end{remark}

\subsection{Globalization}
We now suppose that in period $1$, $H$ and $F$ sign a trade agreement such that capital and consumption goods are freely mobile between the two countries. The Law of One Price yields $p_H=p_F$ and $\rho_H= \rho_F$. Hence $r_H=r_F=r$.
\\[5mm]
Let $\Pi_{i}^G$ denote the maximum profit of the firm in country $i$ ($i=H,F$) under globalization:
$$\Pi_i ^G= \max _{k_{i,1}} \{A_i k_{i,1}^\alpha -r k_{i,1}\}$$
We obtain:
 \begin{eqnarray}
&&\hbox{Demand for capital:  }  k^G _{id,1} = \left (\frac {A_i \alpha }{r}\right ) ^{\frac{1}{1-\alpha}} \label{demand k^H}\\
&&\hbox{Maximal profit:  }\Pi_i^G= \left (\frac{A_i}{r ^\alpha} \right ) ^ {\frac{1}{1-\alpha}} (\alpha ^{\frac{\alpha}{1-\alpha}}- \alpha ^{\frac{1}{1-\alpha}}) \label{pi^H}
\end{eqnarray}
We observe that $\Pi_i ^G$ is positive since $\alpha \in (0,1)$.
The representative consumer of country $i$ solves:
\begin{eqnarray*}
&&\max\{ \ln c_{i,0} +\ln c_{i,1}\}\\
&&\hbox{s.t. } c_{i,0} + k_{i,1} = A_i k_{i,0}^\alpha\\
&&c_{i,1}= \Pi_i^G + r k_{i,1}
\end{eqnarray*}
Tedious calculations give the supply of capital from the consumer:
\begin{eqnarray}
k^G _{is,1}= \frac{1}{2}\left [A_i k_{i,0} ^\alpha + \left (\frac{A_i}{r}\right)^{\frac{1}{1-\alpha}} \left(\alpha ^{\frac{1}{1-\alpha}}- \alpha ^{\frac{\alpha}{1-\alpha}}\right ) \right ] \label{e_kG}
\end{eqnarray}
\begin{proof} See \hyperref[A1]{Appendix~\ref*{A1}}.\end{proof}

\subsubsection{Markets Clearing}
\textbf{(i) Capital Market}\\
When the capital market clears, we have a balance between demand and supply for capital in a global market, i.e.,:
$$k^G_{Hd,1}+ k^G _{Fd,1}= k^G _{Hs,1}+k^G _{Fs,1}$$
which is equivalent to:
\begin{eqnarray*}
&&\left (\frac {A_H \alpha }{r}\right ) ^{\frac{1}{1-\alpha}} + \left (\frac {A_F \alpha }{r}\right ) ^{\frac{1}{1-\alpha}} \\
&& =\frac{1}{2}\left [A_H k_{H,0}^\alpha + \left (\frac{A_H}{r}\right)^{\frac{1}{1-\alpha}} \left(\alpha ^{\frac{1}{1-\alpha}}- \alpha ^{\frac{\alpha}{1-\alpha}}\right ) \right ] + \frac{1}{2}\left [A_F k_{F,0} ^\alpha + \left (\frac{A_F}{r}\right)^{\frac{1}{1-\alpha}} \left(\alpha ^{\frac{1}{1-\alpha}}- \alpha ^{\frac{\alpha}{1-\alpha}}\right ) \right ] \end{eqnarray*}
\\
The common real interest rate for the two countries is:
\begin{eqnarray}\label{e_rvalue}
r^{\frac{1}{1-\alpha} }= %
\left [(1+\alpha) \alpha ^{\frac{\alpha}{1-\alpha}}\right ] \frac{ A_H^{\frac{1}{1-\alpha}} + A_F^{\frac{1}{1-\alpha}}}{A_H k_{H,0} ^\alpha + A_F k_{F,0} ^\alpha }
\end{eqnarray}  
\begin{proof}See \hyperref[A2]{Appendix~\ref*{A2}}.\end{proof} \par

{\raggedright \textbf{(ii) Consumption good market clearing}:}
Once the capital market clears, the consumption good market also clears, according to the Walras’ law. 
\subsubsection{Consequences of Globalization}
We assume that $r_F > r_H$, and this condition is satisfied when $A_F $ is sufficiently larger than $A_H$. Then, we could find: $$	r_H < r < r_F $$
where: $\lim\limits_{A_F\rightarrow +\infty} r_F = +\infty $.

\begin{proposition}\label{p_Af}
There exists a value ${\tilde A}_F$ such that, if $A_F \geq { \tilde A}_F$ then $\Pi^G_{H}<\Pi_{H}$.
\end{proposition}
\begin{proof} Equation (\ref{e_rvalue}) gives us: $\lim\limits_{A_F\rightarrow +\infty}r=+\infty$ and Equation (\ref{pi^H}): $\lim\limits_{r\rightarrow +\infty}\Pi^G_{H}=0$. Hence, when $A_F$ is high enough, we have $\Pi^G_{H}<\Pi_{H}$.\end{proof}

\hyperref[p_Af]{Proposition~\ref*{p_Af}} implies that when country $H$'s TFP $A_H$ is sufficiently low relative to country $F$'s TFP, country $H$ loses from globalization while country $F$ gains. In other words, trade liberalization reduces economic growth in developing countries and increases it in developed countries. In particular, as $A_F$ becomes arbitrarily large, profits in $H$ converge to zero. Consequently, very little production occurs in country $H$, while country $F$ attracts most of the capital from both economies.


\subsection{How a developing country can gain under Globalization}

Notice that under our assumptions, $H$ is a developing country, while $F$ represents the rest of the world (or, in a particular case, the main trading partners of country $H$). Hence, $k_{H,0}$ is very small relative to $k_{F,0}$. If the TFP of $H$, denoted by $A_H$, is not too low, it may be the case that $r_H > r_F$. Under globalization, the equilibrium interest rate then satisfies $r_F < r < r_H$. In this situation, country $H$ benefits from globalization: its output is higher, and it earns greater profits.

Unfortunately, in practice, $A_H$ is often very low, particularly in less developed economies, such as many African countries, so that the condition $r_F < r < r_H$ does not hold. This helps explain why these countries may lose from globalization, as discussed in \hyperref[literature]{Section~\ref*{literature}}. The key question, therefore, becomes how a developing country can achieve higher economic growth under trade liberalization when its TFP is low. In our framework, this involves determining how the TFP of country $H$, $A_H$, can become sufficiently high to satisfy the condition $r_F < r < r_H$. Our proposed solution is to invest in innovation to improve TFP. We now describe the mechanism in detail.

\subsubsection{Under the Autarky Period}

Now, we assume that country $H$ will only use a fraction $\zeta\in (0,1)$ of the initial capital $k_{H,0}$ for its production in the first period. The remaining $(1-\zeta)$ will be used for investment in innovation that improves the TFP in the second period. We assume that the country H's TFP can be now expressed as: 
$$\left [\lambda (1-\zeta) k_{H,0} +1 \right ]^\theta A_H, \lambda >0, \theta >0$$ 
where parameters $\lambda, \theta $ measure the efficiency of the technology used to improve the TFP.\footnote{Several theoretical frameworks express the impact of investment in innovation (or R\&D expenditures) on TFP by providing a TFP function, which is increasing in the level of investment in innovation. Nonetheless, there are different ways to present such a relationship. For example, in \citet{bruno09,levan2010}, the authors rely on a linear function, whereas in \citet{pham25}, the authors use a non-linear function.} The constraints of country $H$, in autarky, now become:
\begin{eqnarray*}
&& c_{H,0} + k_{H,1} = A_H (\zeta k_{H,0})^\alpha \\
&&	c_{H,1}= \left[\lambda (1-\zeta) k_{H,0} +1 \right ]^\theta A_H k_{H,1} ^\alpha 
\end{eqnarray*}
We obtain the new capital supply in period $1$, $\tilde k _{Hs,1}$, and the new capital demand in period $1$, $\tilde k_{Hd,1}$, that are:
\begin{eqnarray*}&&\tilde k _{Hs,1}= \frac{\alpha}{1+\alpha}A_H {(\zeta k_{H,0})}^\alpha \\
&& \tilde k _{Hd,1}= \left( \frac {[\lambda (1-\zeta)  k_{H,0} +1 ]^\theta\alpha A_H}{\tilde r_H}\right ) ^{\frac{1}{1-\alpha}} 
\end{eqnarray*}
The new interest rate under autarky is:
\begin{eqnarray}\label{e_r}
    \tilde r_H= \frac{[\lambda (1-\zeta)  k_{H,0} +1 ]^\theta (A_H)^\alpha \alpha ^\alpha (1+\alpha)^{1-\alpha}}{\left (\zeta k_{H,0}\right )^{\alpha (1-\alpha)}}
\end{eqnarray}
The new maximum profit of country H under autarky is:
\begin{eqnarray}\label{tilde pi^H autarky}
\tilde \Pi_H = \left( \frac{\left[ \lambda (1-\zeta) k_{H,0} +1 \right ]^\theta A_H}{{\tilde r_H} ^\alpha} \right) ^ {\frac{1}{1-\alpha}} (\alpha ^{\frac{\alpha}{1-\alpha}}- \alpha ^{\frac{1}{1-\alpha}}) 
\end{eqnarray}
For country $F$, there is no change, so the real interest rate and the firm profit are as in Equation (\ref{r^H},\ref{pi^H autarky}).

\subsubsection{Under the Globalization Period}
In period 1, under globalization, capital and consumption goods are freely mobile between the two countries. Again, the LOP yields $ \tilde r_H^G=r_F^G=\tilde r $ where $\tilde r$ is the new common interest rate. 

The firm in country $H$ maximizes its profit, such as:
$$ \tilde \Pi_H^G= \max _{\tilde k_{H,1}} \Bigg\{\left[\lambda (1-\zeta) k_{H,0} +1 \right ]^\theta A_H k_{H,1}^\alpha -\tilde r k_{H,1} \Bigg\}$$

We obtain:
 \begin{eqnarray}
&&\hbox{Demand for capital:  }  \tilde k _{Hd,1}= \left( \frac {[\lambda (1-\zeta)  k_{H,0} +1 ]^\theta A_H \alpha}{\tilde r }\right ) ^{\frac{1}{1-\alpha}} \label{tilde demand k^H}\\
&&\hbox{Maximal profit:  } \tilde \Pi_H ^G = \left( \frac{\left[ \lambda (1-\zeta) k_{H,0} +1 \right ]^\theta A_H}{{ \tilde{r }} ^\alpha} \right) ^ {\frac{1}{1-\alpha}} (\alpha ^{\frac{\alpha}{1-\alpha}}- \alpha ^{\frac{1}{1-\alpha}})  \label{tilde pi^H}
\end{eqnarray}
The representative consumer of country $H$ solves:
\begin{eqnarray*}
	&&\max\{ \ln c^G_{H,0} +\ln c^G_{H,1}\}\\
	&&\hbox{s.t. } c^G_{H,0} + \tilde k_{H,1} = A_H (\zeta k_{H,0})^\alpha\\
	&&c^G_{H,1}= \tilde \Pi_H^G + \tilde r \tilde k_{H,1}
\end{eqnarray*}
We obtain the supply of capital from the consumer:
\begin{eqnarray}
\tilde k^G_{Hs,1}= \frac{1}{2}\left [A_H (\zeta k_{H,0}) ^\alpha + \left (\frac{\left[ \lambda (1-\zeta) k_{H,0} +1 \right ]^\theta A_H}{{ \tilde r}}\right)^{\frac{1}{1-\alpha}} \left(\alpha ^{\frac{1}{1-\alpha}}- \alpha ^{\frac{\alpha}{1-\alpha}}\right ) \right ] \label{tilde supply k^H} 
\end{eqnarray}
\begin{proof} See \hyperref[A3]{Appendix~\ref*{A3}}.\end{proof}
{\raggedright \textbf{Capital Market Clearing}}\\
Under globalization, the capital market clears when we get a balance between demand for capital and supply for capital:
$$ \tilde k_{Hd,1}+  k_{Fd,1}= \tilde k_{Hs,1}+  k_{Fs,1}$$
which is equivalent to:
\begin{multline*}
\left (\frac {[\lambda (1-\zeta)  k_{H,0} +1 ]^\theta A_H \alpha }{\tilde r}\right ) ^{\frac{1}{1-\alpha}} + \left (\frac {A_F \alpha }{\tilde r}\right ) ^{\frac{1}{1-\alpha}} \\
=\frac{1}{2}\left [A_H (\zeta k_{H,0}) ^\alpha + \left (\frac{[\lambda (1-\zeta)  k_{H,0} +1 ]^\theta A_H}{\tilde r}\right)^{\frac{1}{1-\alpha}} \left(\alpha ^{\frac{1}{1-\alpha}}- \alpha ^{\frac{\alpha}{1-\alpha}}\right ) \right ] \\
+ \frac{1}{2}\left [A_F (k_{F,0}) ^\alpha + \left (\frac{A_F}{\tilde r}\right)^{\frac{1}{1-\alpha}} \left(\alpha ^{\frac{1}{1-\alpha}}- \alpha ^{\frac{\alpha}{1-\alpha}}\right ) \right ] \end{multline*}
The new common real interest rate for the two countries is:
\begin{eqnarray}\label{e_newr}
\tilde r^{\frac{1}{1-\alpha} }= \left [(1+\alpha) \alpha ^{\frac{\alpha}{1-\alpha}}\right ] \frac{\left ([\lambda (1-\zeta)  k_{H,0} +1 ]^\theta A_H\right)^{\frac{1}{1-\alpha}} + \left (A_F\right)^{\frac{1}{1-\alpha}}}{A_H (\zeta k_{H,0}) ^\alpha + A_F k_{F,0} ^\alpha } 
\end{eqnarray}
\begin{proof}See \hyperref[A4]{Appendix~\ref*{A4}}
\end{proof}

\begin{proposition}\label{p_zeta}
There exists $\tilde\zeta$ such that if $\zeta<\tilde\zeta$ then $\tilde r _H >\tilde r> r_F$. In this case, globalization will lead the capital to flow from $F$ to $H$. 
\end{proposition}

\begin{proof} Equation (\ref{e_r}) give us $\lim\limits_{\zeta\rightarrow 0}\tilde{r}_H=+\infty$. Then when $\zeta$ is small enough, we have $\tilde{r}_H>r_F.$
\end{proof}
\hyperref[p_zeta]{Proposition~\ref*{p_zeta}} indicates that by allocating a sufficient level of investment in innovation, country $H$ will receive a net inflow of capital and thereby have a higher production. Country $H$ gains from globalization through increased production. In other words, country $H$ experiences a higher economic growth with trade liberalization.

In addition, other parameters contribute to increasing the real interest rate $\tilde r_H$, such as the efficiency of the technology used to improve TFP $\lambda$ and $\theta$ or the initial level of TFP $A_H$. In other words, the higher the level of these parameters, the smaller the threshold $\tilde\zeta$ to ensure inequality $\tilde r_H>\tilde r >r_F$. Consequently, countries with weak infrastructure levels or strong corruption, resulting in a high level of uncertainty or adjustment lags, thus low values of $\lambda, \theta$, may have difficulty improving their TFP, except that they dedicate a high savings rate to invest in innovation. By contrast, some externalities, such as FDI spillovers, technological transfer, or foreign aid, can help improve TFP (e.g., \citealp{schiff17, yu22, nguyenhuu24}). In our framework, these factors can be implicitly reflected in a higher value of $A_H$, thereby lowering the threshold $\tilde\zeta$, and consequently facilitating developing countries' benefit from globalization.
\\ 
{\raggedright \textbf{Consumption, Utility in Country H and Globalization:}}\\
The consumption of the representative consumer of country H in each period is given by:
\begin{eqnarray*}
	c^{G^\star}_{H,0} && = A_H (\zeta {k_{H,0}})^\alpha - \tilde  k_{Hs,1} \\
	&& = \frac{1}{2}\left [A_H (\zeta k_{H,0}) ^\alpha - \left (\frac{\left[ \lambda (1-\zeta) k_{H,0} +1 \right ]^\theta A_H}{{ \tilde r}}\right)^{\frac{1}{1-\alpha}} \left(\alpha ^{\frac{1}{1-\alpha}}- \alpha ^{\frac{\alpha}{1-\alpha}}\right ) \right ] \end{eqnarray*} 
\begin{eqnarray*}
	c^{G^\star}_{H,1} && = \tilde \Pi^G_H + \tilde r \tilde k_{Hs,1}\\
	&& = \frac{\tilde r}{2}\left [A_H (\zeta k_{H,0}) ^\alpha - \left (\frac{\left[ \lambda (1-\zeta) k_{H,0} +1 \right ]^\theta A_H}{{ \tilde r}}\right)^{\frac{1}{1-\alpha}} \left(\alpha ^{\frac{1}{1-\alpha}}- \alpha ^{\frac{\alpha}{1-\alpha}}\right ) \right ] \end{eqnarray*} 
and from that, we also see: $ c^{G^\star}_{H,1} = \tilde{r} c^{G^\star}_{H,0}$ \\[3mm]
The utility of country H with innovation investment under globalization is given by:
\begin{eqnarray*}
 \tilde U^G_H &&= \ln c^{G^\star}_{H,0} +\ln c^{G^\star}_{H,1} \\
 &&= \ln \Bigg\{ \frac{\tilde r}{2}\left [A_H (\zeta k_{H,0}) ^\alpha - \left (\frac{\left[ \lambda (1-\zeta) k_{H,0} +1 \right ]^\theta A_H}{{ \tilde r}}\right)^{\frac{1}{1-\alpha}} \left(\alpha ^{\frac{1}{1-\alpha}}- \alpha ^{\frac{\alpha}{1-\alpha}}\right ) \right ]^2 \Bigg\}
\end{eqnarray*}
and that without technology investment:
\begin{eqnarray}
   U^G_H  = \ln \Bigg\{ \frac{ r}{2}\left [A_H ( k_{H,0}) ^\alpha - \left (\frac{ A_H}{{  r}}\right)^{\frac{1}{1-\alpha}} \left(\alpha ^{\frac{1}{1-\alpha}}- \alpha ^{\frac{\alpha}{1-\alpha}}\right ) \right ]^2 \Bigg\}
\end{eqnarray}
where:
\begin{eqnarray*}
	&& \tilde r^{\frac{1}{1-\alpha} }= \left [(1+\alpha) \alpha ^{\frac{\alpha}{1-\alpha}}\right ] \frac{\left ([\lambda (1-\zeta)k_{H,0} +1 ]^\theta A_H\right)^{\frac{1}{1-\alpha}} + \left (A_F\right)^{\frac{1}{1-\alpha}}}{A_H (\zeta k_{H,0}) ^\alpha + A_F k_{F,0} ^\alpha }\\
		&&  r^{\frac{1}{1-\alpha} }= \left [(1+\alpha) \alpha ^{\frac{\alpha}{1-\alpha}}\right ] \frac{\left ( A_H\right)^{\frac{1}{1-\alpha}} + \left (A_F\right)^{\frac{1}{1-\alpha}}}{A_H k_{H,0}^\alpha + A_F k_{F,0}^\alpha }
\end{eqnarray*}
note that $ \tilde{r} > r $ and when $ \zeta \rightarrow 1 $ then $ \tilde U^H \rightarrow U^H $. 
\\ \\
Now, our question is whether country $H$ can gain under globalization in terms of utility. That is if the country saves a portion of its initial capital (hence less consumption $ c_{H,0} $) to invest in technology/innovation for improving the TFP in the second period (hence more consumption $ c_{H,1} $), its associated utility will exceed that in the case without investment, i.e. $ \tilde U^G_H > U^G_H $. However, the result could not be clear because these utilities are affected by different factors, such as capital elasticity, $\alpha$; efficiency of technology used to improve the TFP, $\theta$ and $\lambda$; share of investment in capital, $\zeta$; or TFP level of country $F$, $A_F$. 
\section{Simulations}

To illustrate our theoretical framework, particularly \hyperref[p_zeta]{Proposition~\ref*{p_zeta}}, we consider a simple numerical calibration of the model. The parameter values are set as follows:
\[
\alpha=\frac{1}{3}, \qquad \theta=0.3, \qquad \lambda=0.5, \qquad A_H=2, \qquad A_F=6, \qquad k_{H,0}=1, \qquad k_{F,0}=4.
\]

These values are chosen for illustrative purposes and are meant to capture the situation of a developing economy that enters globalization with lower productivity and a lower capital stock than the rest of the world. The capital share parameter $\alpha=\frac{1}{3}$ is standard in the growth literature. The parameter $\theta=0.3$ implies that innovation improves productivity, but with diminishing returns. The parameter $\lambda=0.5$ reflects a moderate efficiency level of investment in innovation. Finally, Assumptions $A_H<A_F$ and $k_{H,0}<k_{F,0}$ reflect an initial technological and capital gap between the developing country and the foreign country. Let us define $i=1-\zeta$, the share of initial endowment (i.e., $k_{H,0}$) allocated to innovation. 

We then consider two alternative cases. In the first, foreign productivity is higher, with:
\[
A_F=8,
\]
while all other parameters remain unchanged. In the second, the productivity effect of innovation is weaker, with:
\[
\theta=0.15,
\]
again keeping the other parameters constant. These two alternative analyses allow us to assess how the critical threshold of innovation investment responds to a larger foreign productivity advantage and to a lower efficiency of domestic innovation. As a consequence, we have a more in-depth understanding of our discussion below \hyperref[p_zeta]{Proposition~\ref*{p_zeta}}.

For each calibration, we compute three interest rates as a function of the share of initial capital allocated to innovation, $i$, varying between 0 and 1:

\begin{itemize}
    \item[(i)] the domestic interest rate in the developing country with innovation, $\tilde r_H$;
    \item[(ii)] the world interest rate under globalization with innovation, $\tilde r$;
    \item[(iii)] the foreign country's interest rate, $r_F$, which is independent of $i$.
\end{itemize}

The baseline simulation, illustrated in \hyperref[f_pro2_baseline]{Figure~\ref*{f_pro2_baseline}}, yields a threshold level given in \hyperref[p_zeta]{Proposition~\ref*{p_zeta}}
\[
i^* \approx 0.1508
\]
such that:
\[
r_F > \tilde r_H \qquad \text{for} \qquad i<i^*,
\]
whereas
\[
\tilde r_H > r_F \qquad \text{for} \qquad i>i^*.
\]

\begin{figure}[!htbp]
    \centering
    \includegraphics[width=0.75\textwidth]{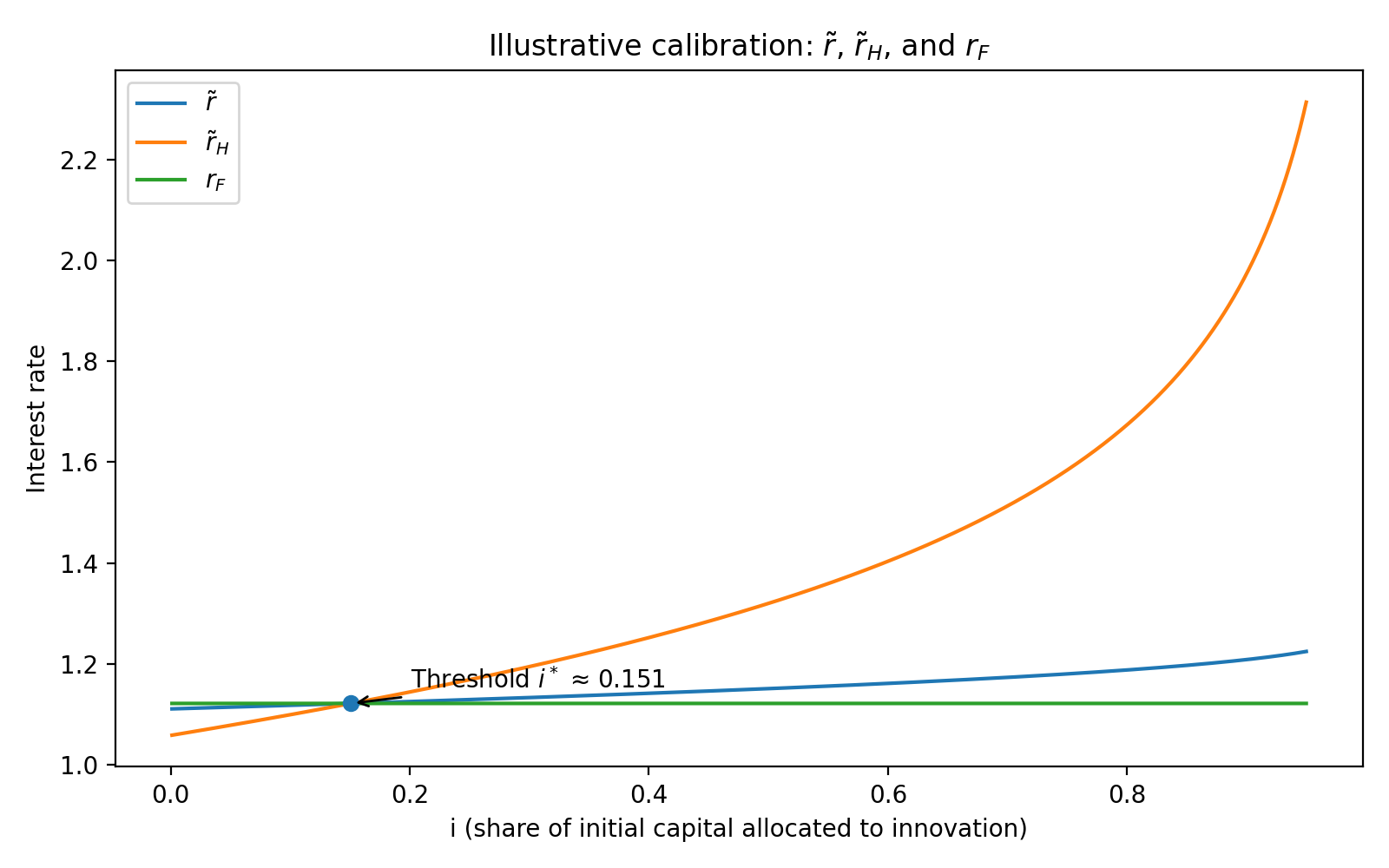}
    \caption{Illustrative calibration of \hyperref[p_zeta]{Proposition~\ref*{p_zeta}}: comparison between $r_F$ and $\tilde r_H$}
    \label{f_pro2_baseline}
\end{figure}

Accordingly, for low levels of innovation investment, the foreign interest rate remains above the domestic one. In this case, the developing country enters globalization from a position of technological weakness and is unlikely to benefit fully from international integration. By contrast, once the share of investment in innovation exceeds the threshold $i^*$, domestic productivity improves sufficiently for the domestic interest rate to overtake the foreign one. 

\hyperref[f_pro2_alter]{Figure~\ref*{f_pro2_alter}} reports the two alternative calibrations. Panel \subref{f_pro2_highAf} considers the case of higher foreign productivity and shows that the threshold increases to
\[
i^* \approx 0.3709,
\]
indicating that a larger foreign productivity advantage makes it significantly harder for the developing country to benefit from globalization. Panel \subref{f_pro2_lowtheta} presents the case with lower innovation efficiency and shows that the threshold becomes
\[
i^* \approx 0.1833,
\]
implying that weaker innovation efficiency also raises the level of investment required for globalization to generate positive returns.

\begin{figure}[!htbp]
    \centering
    \begin{subfigure}[b]{0.48\textwidth}
        \centering
        \includegraphics[width=\textwidth]{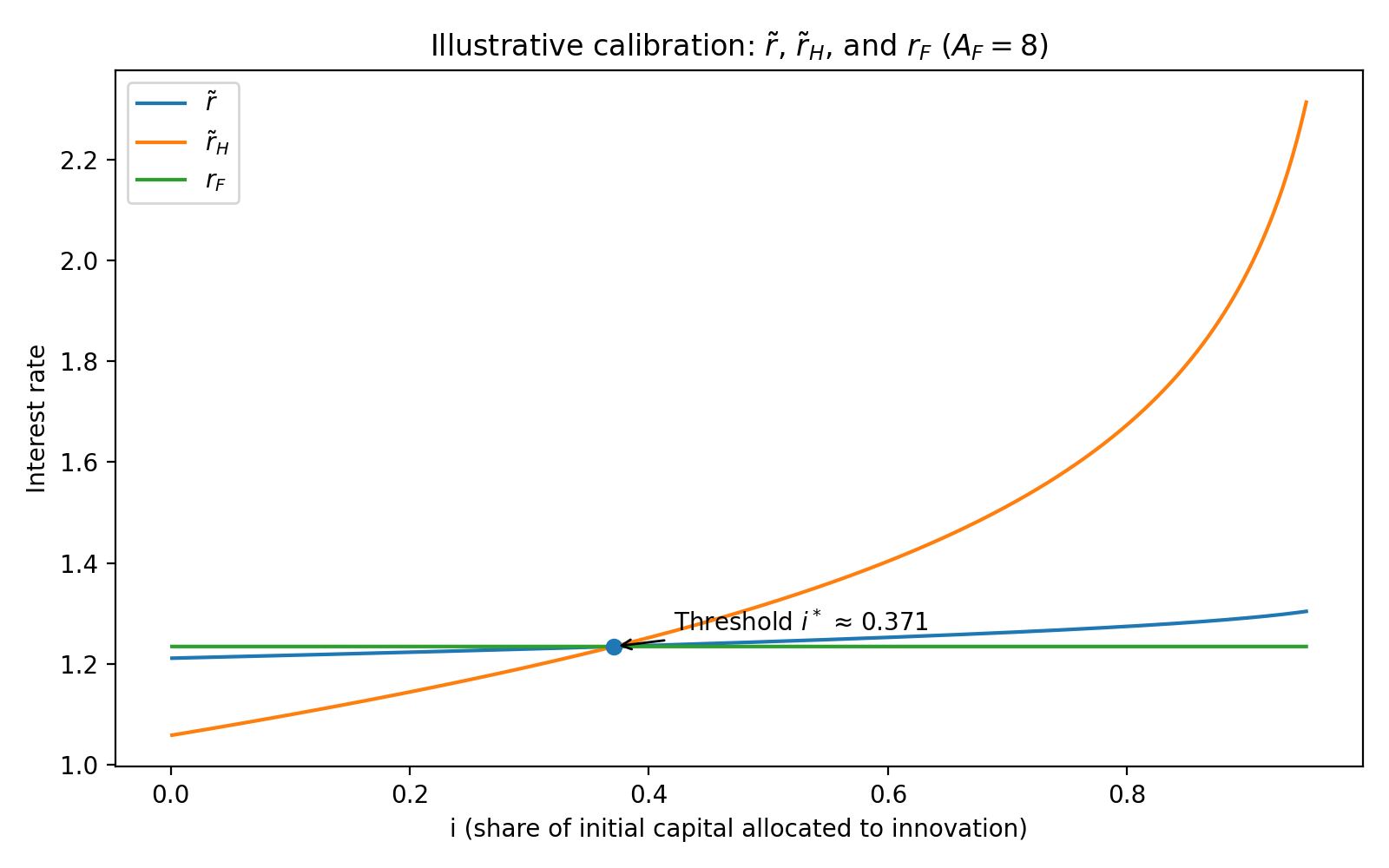}
        \caption{Higher foreign productivity ($A_F=8$).}
        \label{f_pro2_highAf}
    \end{subfigure}
    \hfill
    \begin{subfigure}[b]{0.48\textwidth}
        \centering
        \includegraphics[width=\textwidth]{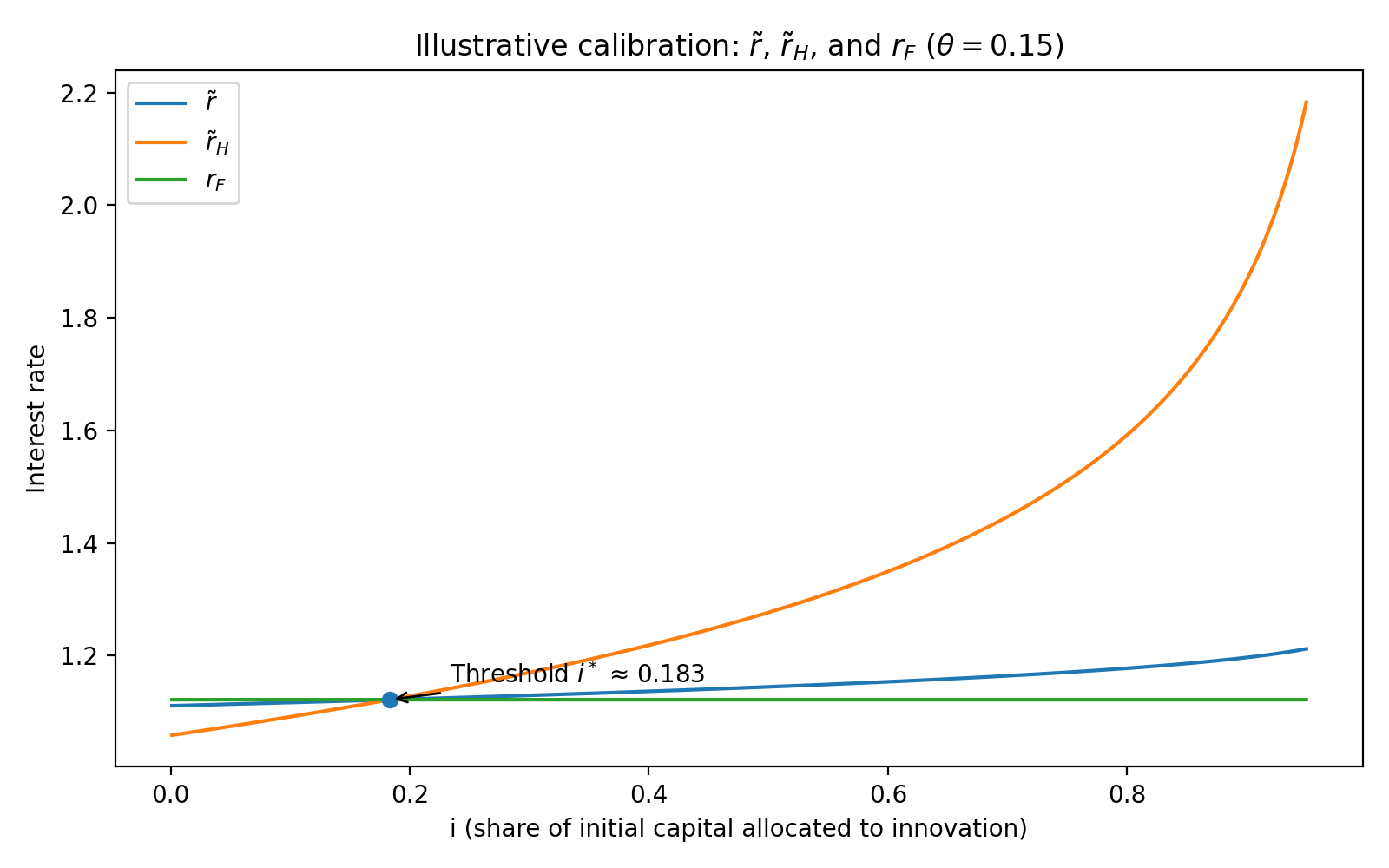}
        \caption{Lower innovation efficiency ($\theta=0.15$).}
        \label{f_pro2_lowtheta}
    \end{subfigure}
    \caption{Alternative calibrations: comparison between $\tilde r$, $\tilde r_H$, and $r_F$.}
    \label{f_pro2_alter}
\end{figure}

Together, these calibrations provide a clear numerical illustration of \hyperref[p_zeta]{Proposition~\ref*{p_zeta}}. In all cases, the gains from globalization are conditional on a sufficiently high level of innovation investment. Moreover, the threshold itself depends on the economy's structural characteristics. A larger initial productivity gap against the foreign economy, captured here by a higher value of $A_F$, substantially increases the threshold. Similarly, a lower innovation effectiveness, captured by a smaller value of $\theta$, also makes it more difficult for the developing country to take advantage of globalization. These results reinforce the model's central message: globalization is not automatically growth-enhancing for developing economies. Its benefits critically depend on domestic innovative capacity and on the country’s initial position relative to the rest of the world.\footnote{Notice that when $A_F$ becomes too high or $\theta$ remains too low, threshold $i^*$ can exceed 1. In this case, globalization is no longer beneficial to the developing country, except that it can borrow money from abroad.}

The policy implications are straightforward. First, trade liberalization should not be viewed as a substitute for domestic efforts to invest in innovation. Second, when the technological gap with the rest of the world is large, the effort required in innovation is substantially greater. Third, the effectiveness of innovation policy itself matters: weak innovation systems make it harder for developing countries to cross the threshold needed to benefit from openness. More generally, these numerical illustrations suggest that globalization is most beneficial when accompanied by sustained investment in innovation, technology adoption, and productivity improvement.

\section{Conclusion}\label{con}
How can a developing country benefit from trade liberalization? In this paper, we contribute to the trade-growth literature by examining a two-period model and demonstrating how a developing country can benefit from globalization. We show that a developing country loses from globalization when its TFP is very low relative to that of the rest of the world. We also show that a low-TFP country can still benefit from trade by investing in innovation. Nevertheless, to do so, the level of investment must exceed a critical threshold. From a policy perspective, our theoretical results suggest that trade liberalization should not be viewed in isolation. An investment in innovation to improve TFP is also required to achieve higher economic growth.

Although our results are meaningful, several caveats deserve attention for future research. First, the highly stylized nature of the model could be relaxed by weakening the law of one price (LOP), which we impose here. In reality, arbitrage is often limited by transport costs, tariffs, and market power exercised by large buyers or sellers. Second, in our framework, globalization entails only the mobility of capital and goods. In practice, globalization also involves the transfer of technology and technological spillovers from developed to developing countries. It is therefore important to examine how these channels affect the dynamic evolution of TFP in developing economies. To this end, we need to extend the model to an infinite-horizon growth framework.


\appendix
\section{Appendix}
\subsection{Proof of Equation (\ref{e_kG})}\label{A1}
The representative consumer of country $i$ solves:
\begin{eqnarray*}
	&&\max\{ \ln c_{i,0} +\ln c_{i,1}\}\\
	&&\hbox{s.t. } c_{i,0} + k_{i,1} = A_i k_{i,0}^\alpha\\
	&&c_{i,1}= \Pi_i^G + r k_{i,1}
\end{eqnarray*}
The equivalent problem:
\begin{eqnarray*}
	\max \Big\{\ln(A_i k_{i,0}^\alpha - k_{i,1}) + \ln(\Pi_i^G + r k_{i,1})\Big\}
\end{eqnarray*}
First Order Condition with respect to $ k_{i,1}$:
\begin{eqnarray*}
	&& - \frac{1}{A_i k_{i,0}^\alpha - k_{i,1}} + \frac{r}{\Pi_i ^G + r k_{i,1}} = 0\\ 
	&& \Leftrightarrow \frac{1}{A_i k_{i,0}^\alpha - k_{i,1}} = \frac{r}{\Pi_i ^G + r k_{i,1}}\\ 
	&& \Leftrightarrow \Pi_i ^G + r k_{i,1} = r A_i k_{i,0}^\alpha - r k_{i,1} \\
	&& \Leftrightarrow k_{i,1} = \frac{1}{2} \left( A_i k_{i,0}^\alpha - \frac{1}{r} \Pi_i ^G \right) \\
	&& \Leftrightarrow k_{i,1} = \frac{1}{2} \left( A_i k_{i,0}^\alpha - \frac{1}{r} \left( \frac{A_i}{r ^\alpha} \right ) ^ {\frac{1}{1-\alpha}} (\alpha ^{\frac{\alpha}{1-\alpha}}- \alpha ^{\frac{1}{1-\alpha}}) \right)
\end{eqnarray*}
We obtain:
\begin{eqnarray*}
k^G _{is,1}= \frac{1}{2}\left [A_i k_{i,0} ^\alpha + \left (\frac{A_i}{r}\right)^{\frac{1}{1-\alpha}} \left(\alpha ^{\frac{1}{1-\alpha}}- \alpha ^{\frac{\alpha}{1-\alpha}}\right ) \right ]
\end{eqnarray*}
\subsection{Proof of Equation (\ref{e_rvalue})}\label{A2}

\begin{align*}
& \left(\frac{A_H \alpha}{r}\right)^{\frac{1}{1-\alpha}} + \left(\frac{A_F \alpha}{r}\right)^{\frac{1}{1-\alpha}} \\
& \qquad = \frac{1}{2}\left[A_H k_{H,0}^\alpha + \left(\frac{A_H}{r}\right)^{\frac{1}{1-\alpha}}\left(\alpha^{\frac{1}{1-\alpha}}-\alpha^{\frac{\alpha}{1-\alpha}}\right)\right] + \frac{1}{2}\left[A_F k_{F,0}^\alpha + \left(\frac{A_F}{r}\right)^{\frac{1}{1-\alpha}}\left(\alpha^{\frac{1}{1-\alpha}}-\alpha^{\frac{\alpha}{1-\alpha}}\right)\right] \\
\intertext{by simple calculations:}
\Leftrightarrow{}\ & \left(\frac{A_H \alpha}{r}\right)^{\frac{1}{1-\alpha}} + \left(\frac{A_F \alpha}{r}\right)^{\frac{1}{1-\alpha}} - \frac{1}{2}\left(\frac{A_H}{r}\right)^{\frac{1}{1-\alpha}}\left(\alpha^{\frac{1}{1-\alpha}}-\alpha^{\frac{\alpha}{1-\alpha}}\right) - \frac{1}{2}\left(\frac{A_F}{r}\right)^{\frac{1}{1-\alpha}}\left(\alpha^{\frac{1}{1-\alpha}}-\alpha^{\frac{\alpha}{1-\alpha}}\right) \\
& \qquad = \frac{1}{2} A_H k_{H,0}^\alpha + \frac{1}{2} A_F k_{F,0}^\alpha \\
\Leftrightarrow{}\ & \frac{1}{2}\left(\frac{A_H}{r}\right)^{\frac{1}{1-\alpha}}\left(\alpha^{\frac{1}{1-\alpha}}+\alpha^{\frac{\alpha}{1-\alpha}}\right) + \frac{1}{2}\left(\frac{A_F}{r}\right)^{\frac{1}{1-\alpha}}\left(\alpha^{\frac{1}{1-\alpha}}+\alpha^{\frac{\alpha}{1-\alpha}}\right) = \frac{1}{2} A_H k_{H,0}^\alpha + \frac{1}{2} A_F k_{F,0}^\alpha \\
\Leftrightarrow{}\ & \left[\left(\frac{A_H}{r}\right)^{\frac{1}{1-\alpha}} + \left(\frac{A_F}{r}\right)^{\frac{1}{1-\alpha}}\right]\left(\alpha^{\frac{1}{1-\alpha}}+\alpha^{\frac{\alpha}{1-\alpha}}\right) = A_H k_{H,0}^\alpha + A_F k_{F,0}^\alpha \\
\Leftrightarrow{}\ & \frac{\left[\left(A_H\right)^{\frac{1}{1-\alpha}} + \left(A_F\right)^{\frac{1}{1-\alpha}}\right]\left[(1+\alpha)\alpha^{\frac{\alpha}{1-\alpha}}\right]}{r^{\frac{1}{1-\alpha}}} = A_H k_{H,0}^\alpha + A_F k_{F,0}^\alpha \\
\Leftrightarrow{}\ & r^{\frac{1}{1-\alpha}} = \left[(1+\alpha)\alpha^{\frac{\alpha}{1-\alpha}}\right]\frac{\left(A_H\right)^{\frac{1}{1-\alpha}} + \left(A_F\right)^{\frac{1}{1-\alpha}}}{A_H k_{H,0}^\alpha + A_F k_{F,0}^\alpha}
\end{align*}

\subsection{Proof of Equation (\ref{tilde supply k^H})}\label{A3}
Problem:
\begin{eqnarray*}
	&&\max\{ \ln c^G_{H,0} +\ln c^G_{H,1}\}\\
	&&\hbox{s.t. } c^G_{H,0} + \tilde k_{Hs,1} = A_H {(\zeta k_{H,0})}^\alpha\\
	&&c^G_{H,1}= \tilde \Pi^G_H + \tilde r \tilde k_{Hs,1}
\end{eqnarray*}
The equivalent one:
\begin{eqnarray*}
	\max \Big\{\ln(A_H {(\zeta k_{H,0})}^\alpha - \tilde k_{Hs,1}) + \ln(\tilde \Pi^G_H + \tilde r \tilde k_{Hs,1})\Big\}
\end{eqnarray*}
First Order Condition with respect to $ \tilde k_{Hs,1} $:
\begin{eqnarray*}
	&& - \frac{1}{A_H {(\zeta k_{H,0})}^\alpha - \tilde k_{Hs,1}} + \frac{\tilde r}{\tilde \Pi^G_H + \tilde r \tilde k_{Hs,1}} = 0\\ 
	&& \Leftrightarrow \frac{1}{A_H {(\zeta k_{H,0})}^\alpha - \tilde k_{Hs,1}} =  \frac{\tilde r}{\tilde \Pi^G_H + \tilde r \tilde k_{Hs,1}}\\ 
	&& \Leftrightarrow \tilde k_{Hs,1} = \frac{1}{2} \left( A_H {(\zeta k_{H,0})}^\alpha - \frac{1}{\tilde r} \tilde \Pi^G_H \right) \\
	&& \Leftrightarrow \tilde k_{Hs,1} = \frac{1}{2} \left[ A_H {(\zeta k_{H,0})}^\alpha -  \left( \frac{\left[ \lambda (1-\zeta) k_{H,0} +1 \right ]^\theta A_H}{{ \tilde r}} \right) ^ {\frac{1}{1-\alpha}} (\alpha ^{\frac{\alpha}{1-\alpha}}- \alpha ^{\frac{1}{1-\alpha}})  \right]
\end{eqnarray*}
\subsection{Proof of Equation (\ref{e_newr})}\label{A4}

\begin{align*}
& \tilde k_{Hd,1} + k_{Fd,1} = \tilde k_{Hs,1} + k_{Fs,1} \\
\Leftrightarrow{}\ & \left(\frac{[\lambda(1-\zeta)k_{H,0}+1]^\theta A_H \alpha}{\tilde r}\right)^{\frac{1}{1-\alpha}} + \left(\frac{A_F \alpha}{\tilde r}\right)^{\frac{1}{1-\alpha}} \\
& \qquad = \frac{1}{2}\left[A_H(\zeta k_{H,0})^\alpha + \left(\frac{[\lambda(1-\zeta)k_{H,0}+1]^\theta A_H}{\tilde r}\right)^{\frac{1}{1-\alpha}}\left(\alpha^{\frac{1}{1-\alpha}}-\alpha^{\frac{\alpha}{1-\alpha}}\right)\right] \\
& \qquad\quad + \frac{1}{2}\left[A_F k_{F,0}^\alpha + \left(\frac{A_F}{\tilde r}\right)^{\frac{1}{1-\alpha}}\left(\alpha^{\frac{1}{1-\alpha}}-\alpha^{\frac{\alpha}{1-\alpha}}\right)\right] \\
\intertext{by simple calculations:}
\Leftrightarrow{}\ & \left(\frac{[\lambda(1-\zeta)k_{H,0}+1]^\theta A_H \alpha}{\tilde r}\right)^{\frac{1}{1-\alpha}} + \left(\frac{A_F \alpha}{\tilde r}\right)^{\frac{1}{1-\alpha}} \\
& \qquad - \frac{1}{2}\left(\frac{[\lambda(1-\zeta)k_{H,0}+1]^\theta A_H}{\tilde r}\right)^{\frac{1}{1-\alpha}}\left(\alpha^{\frac{1}{1-\alpha}}-\alpha^{\frac{\alpha}{1-\alpha}}\right) - \frac{1}{2}\left(\frac{A_F}{\tilde r}\right)^{\frac{1}{1-\alpha}}\left(\alpha^{\frac{1}{1-\alpha}}-\alpha^{\frac{\alpha}{1-\alpha}}\right) \\
& \qquad = \frac{1}{2} A_H(\zeta k_{H,0})^\alpha + \frac{1}{2} A_F k_{F,0}^\alpha \\
\Leftrightarrow{}\ & \left[\left(\frac{[\lambda(1-\zeta)k_{H,0}+1]^\theta A_H}{\tilde r}\right)^{\frac{1}{1-\alpha}} + \left(\frac{A_F}{\tilde r}\right)^{\frac{1}{1-\alpha}}\right]\left(\alpha^{\frac{1}{1-\alpha}}+\alpha^{\frac{\alpha}{1-\alpha}}\right) = A_H(\zeta k_{H,0})^\alpha + A_F k_{F,0}^\alpha \\
\Leftrightarrow{}\ & \frac{\left[\left([\lambda(1-\zeta)k_{H,0}+1]^\theta A_H\right)^{\frac{1}{1-\alpha}} + \left(A_F\right)^{\frac{1}{1-\alpha}}\right]\left[(1+\alpha)\alpha^{\frac{\alpha}{1-\alpha}}\right]}{\tilde r^{\frac{1}{1-\alpha}}} = A_H(\zeta k_{H,0})^\alpha + A_F k_{F,0}^\alpha \\
\Leftrightarrow{}\ & \tilde r^{\frac{1}{1-\alpha}} = \left[(1+\alpha)\alpha^{\frac{\alpha}{1-\alpha}}\right]\frac{\left([\lambda(1-\zeta)k_{H,0}+1]^\theta A_H\right)^{\frac{1}{1-\alpha}} + \left(A_F\right)^{\frac{1}{1-\alpha}}}{A_H(\zeta k_{H,0})^\alpha + A_F k_{F,0}^\alpha}
\end{align*}

\printbibliography[heading=bibintoc]

\end{document}